\documentclass[
reprint,
superscriptaddress,
showpacs,
nofootinbib,
amsmath,amssymb,
aps,
pra,
floatfix,longbibliography
]{revtex4-2}
\pdfoutput=1

\usepackage[T1]{fontenc}

\usepackage{graphicx}
\usepackage{dcolumn}
\usepackage[hypertexnames, breaklinks=true, bookmarksnumbered=true, bookmarksopen=true, colorlinks=true, linktocpage=true, citecolor=blue, urlcolor=magenta, linkcolor=magenta]{hyperref}

\usepackage{amsfonts,amsthm}
\usepackage{mathrsfs}
\usepackage{bbm} 
\usepackage{bm} 
\usepackage{mathtools} 

\usepackage[stretch=10]{microtype} 
\usepackage[normalem]{ulem} 
\usepackage{empheq}
\usepackage{color}
\usepackage{tcolorbox}

\usepackage{pxfonts,txfonts}
\usepackage[font=small,justification=raggedright]{caption}
\usepackage{subcaption}

\allowdisplaybreaks

\newcommand{\e}{\bm{e}}

\newcommand{\bra}[1]{\langle #1|}
\newcommand{\ket}[1]{|#1\rangle}

\newcommand{\ip}[2]{\langle #1|#2\rangle}
\newcommand{\op}[2]{|#1\rangle \langle #2|}
\newcommand{\Tr}{\textrm{Tr}}

\newcommand{\blue}[1]{{\leavevmode\color{blue}#1}}

\newtheorem{lemma}{\emph{Lemma}}

\newtheorem{corollary}{Corollary}

\allowdisplaybreaks

\begin{document}


\title{Coherence-Assisted Entanglement Activation During Open Evolution}
\author{Masaya Takahashi}
\email{masayatakahashi.work@gmail.com}
\affiliation{Center for Theoretical Physics, Polish Academy of Sciences,  Lotników 32/46, 02-668 Warszawa, Poland}
\author{Jarosław K. Korbicz}
\affiliation{Center for Theoretical Physics, Polish Academy of Sciences, Lotników 32/46, 02-668 Warszawa, Poland}
\date{February 20, 2025}
\begin{abstract} 
The fact that in the course of an open evolution there is, in a vast majority of cases, the creation of system-environment entanglement out of the initial coherence, has been known for many years. However, how much entanglement can be created under general assumptions has not been analyzed yet. In this work, we investigate system-environment entanglement generation under a broad and important class of interactions known as pure dephasing interactions and for arbitrary initial conditions. Our main results are the following bounds on the relative entropy of entanglement $E_r(\sigma_{SE})$ after evolution:  
$ C_r(\rho_S) - H(\mathcal{I}|\mathcal{M}) \leq E_r(\sigma_{SE}) \leq C_r(\rho_S) - C_r(\sigma_S) + H(\mathcal{I})$,
connecting it on one side to the initial relative entropy of coherence, $C_r(\rho_S)$, and information extractable from the environment, $H(\mathcal{I}|\mathcal{M})$,  and on the other side to the decoherence strength, $ C_r(\rho_S) - C_r(\sigma_S)$, making quantitatively precise the statement that larger entanglement leads to more decoherence. 
As an illustration, we apply our results to the spin-boson model, where a qubit interacts with a bosonic bath, showing time evolution of the entanglement and the Markovian/non-Markovian crossover. Our results help better understand and quantify the complicated correlations, produced during an open evolution, which will hopefully lead to better and new decoherence mitigation techniques.   
\end{abstract}

\maketitle

\section{Introduction}\label{Sec:Introduction}

Entanglement~\cite{Horodecki-2009a} has been the subject of a significant amount of research in the last decades due to its significance not only for quantum information and computation, but for the foundations of quantum mechanics as well. 
Understandably, the vast majority of entanglement studies focus on the purposeful entanglement generation, manipulation, and preservation. However, there are processes where entanglement can be generated or activated spontaneously.  
Open system dynamics (see e.g. \cite{Schlosshauer-2007a,Breuer-2002a} for a modern introduction) is one such example, where the interaction between the system and its environment leads in general to entanglement creation out of initial system coherences. 

Until recently \cite{Roszak-2015a, Roszak-2018a, Roszak-2020a, Rzepkowski-2020a}, this system-environment entanglement had not been systematically studied using modern methods of quantum information theory,
having been mostly treated as a mere source of decoherence. However, with the emergence and the gradual acceptance of a new paradigm in open systems, given by the quantum Darwinism idea \cite{Zurek-2009a} and its refinements such as Spectrum Broadcast Structures and strong Darwinism  (see e.g. \cite{Korbicz-2021a} for a review), the environment has been recognized as a carrier of valuable information about the system. Therefore, studying system-environment entanglement is important not only for understanding decoherence processes, 
since entanglement in general increases decoherence efficiency, see e.g. \cite{Takahashi-2022a} for recent take on this question, but also to know how environment acquires system-related information. 

One can illustrate the system-environment entanglement activation on a paradigmatic example of the so-called pure dephasing evolution, also known as controlled-unitary evolution, alternatively called also a measurement-type evolution:
\begin{align}
    \ket{s_i}\ket{E_0} \rightarrow \ket{s_i}\ket{E_i},
\end{align}
in some preferred basis of the system $ \ket{s_i}$, called the pointer basis ~\cite{Zurek-1981a,Zurek-1982a,Zurek-2003a}. This basis is specified by the nature of the interaction, e.g. the dipole approximation for the electromagnetic interaction, leading to the spin-boson model \cite{Schlosshauer-2007a, Breuer-2002a}.
Whenever there is a coherence present in the initial state, $\ket{\psi_S}=\sum_i a_i\ket{s_i}$, the resulting state $\ket{\psi'_{SE}}$ is obviously entangled:
\begin{align}\label{easy}
\ket{\psi_S}\ket{E_0} \rightarrow \ket{\psi'_{SE}} = \sum_i a_i\ket{s_i} \ket{E_i}
\end{align}
while the reduced subsystem $\Tr_E[\op{\psi'_{SE}}{\psi'_{SE}}]$ lost the initial coherence in $\ket{\psi_S}$.
In the modern language of the resource theories \cite{Chitambar-2019a}, this is an example of the resource interconversion, where one resource, coherence in the pointer basis, gets interconverted into another - entanglement ~\cite{Piani-2011a,Chuan-2012a,Streltsov-2015a,Killoran-2016a,Qiao-2018a,Regula_2018}. Let us remark here that an abstract coherence theory naturally applies to open systems with pure dephasing evolution as the latter contains a natural, preferred basis.

Although the pure state interconversion \eqref{easy} appears simple, the situation gets much more complicated when considering mixed states, especially if one wants to quantify the amount of the generated entanglement. 
Assuming that the initial state of the system is pure, Roszak et al. ~\cite{Roszak-2015a, Roszak-2018a} recently found necessary and sufficient separability conditions for the system-environment split and have developed an operational 
measure of the system-environment entanglement \cite{Roszak-2020a,Rzepkowski-2020a}. However, it is still unknown how to generalize their methods to mixed initial states of the system. 

We introduce some terminology for readers not familiar with quantum coherence. Roughly speaking, an \emph{incoherent state} is simply a state that
is diagonal with respect to some reference basis, known as the
\emph{incoherent basis}, and an \emph{incoherent operation} is a
quantum operation that maps incoherent states to incoherent states. 
Detailed definitions will be reviewed later. Previous studies of the conversion from coherence to entanglement showed that the  entanglement generated by incoherent operations between coherent state in the system and incoherent pure state in environment is upper bounded by the initial coherence in the system with respect to relative entropy measure~\cite{Streltsov-2015a}. 
However, for given Hamiltonian, the corresponding pure dephasing evolution is not incoherent operation generally and we also take the environmental state as arbitrary, their finding does not hold in our case. In pure dephasing evolution, therefore, it is not clear how entanglement generation is quantified through the initial coherence or decoherence in the system.

In this work, we study entanglement generation for mixed states of the system. We derive a lower bound and an upper bound for the system-environment entanglement as measured by the relative entropy of entanglement $E_r$~\cite{Vedral-1997a,Vedral-1998a}.
Following the simple intuition \eqref{easy} and its developments \cite{Piani-2011a,Chuan-2012a,Streltsov-2015a,Killoran-2016a,Qiao-2018a,Regula_2018}, the lower bound is given in terms of the initial coherence of the system with respect to the natural pointer basis defined by the pure dephasing evolution, diminished by the remaining uncertainty about the system after the best measurement on the environment.

The upper bound is quantified with the strength of decoherence in the system, following the intuition that entanglement generation requires not only the initial coherence but the interaction with the environment followed by decoherence which is also suggested by \eqref{easy}. 

Regarding the first lower bound, for two-dimensional central systems, i.e. for qubits, we derive a stronger result in given by a closed formula 
that links entanglement with the distinguishability of the environment states.  
In particular, this allows us to study entanglement dynamics for open qubits, understood as the real-time monitoring of the minimum amount of entanglement produced during the evolution.   
As an illustration, we apply our results to  the spin-boson model (see e.g. \cite{Schlosshauer-2007a} and the references therein). 
We study the dependence of the minimum generated entanglement on various model parameters e.g. the temperature and Ohmicity. Varying the latter through the well-known Markovian/non-Markovian transition, we see the traces of non-Markovianity in the dynamics of entanglement; cf. \cite{Rivas-2010a}.


\section{Preliminaries}\label{Sec:Preliminary}
We consider here the following Hamiltonian: 
\begin{align}\label{eq:Hamiltonian}
    H_{tot} = H_S  +  H_E + \sum_{i=0}^{N-1} \op{i}{i} \otimes V_i
\end{align}
where $H_S$ is the self-Hamiltonian of the system which we assume to commute with the interaction, i.e. $H_S = \sum_i \varepsilon_i \op{i}{i}$. The whole resulting Hamiltonian $H_{tot}$ then takes the so-called pure dephasing form: $H_{tot}= \sum_i \ket i \bra i \otimes (\varepsilon_i +V_i +H_E) \equiv \sum_i \ket i \bra i \otimes V_i'$, with $V_i'\equiv \varepsilon_i +V_i +H_E$.
This Hamiltonian preserves the pointer basis, i.e. $[\ket i \bra i, H_{tot}]=0$, thus affecting only the off-diagonal terms, i.e. the coherence, via the dephasing noise, hence the name. 
We make this choice for simplicity, since inclusion of a general $H_S$ complicates the dynamics immensely, which forces to use approximate methods such as master equations \cite{Schlosshauer-2007a,Breuer-2002a} and ultimately the numerical methods. This type of a much more advanced and complicated analysis is out of the scope of the current work and is left for a future research. Here, $H_E$ and $V_i$'s are arbitrary environmental observables. Eq. \eqref{eq:Hamiltonian} is the most general Hamiltonian that generates evolution of the type \eqref{easy}. 
 The unitary evolution $U$ is given by:
\begin{align}\label{eq:unitary}
    U(t) = \sum_i \op{i}{i} \otimes U_E^i(t), \quad U_E^i(t) = e^{-j(\varepsilon_i+H_E+V_i)t},
\end{align}
where $j$ is imaginary unit. This unitary is also called the controlled-unitary or measurement-type operation in the quantum information terminology. The basis $\ket{i}$ is the pointer basis.
Since observables $V_i$ are arbitrary, in general non-commuting, there does not exist a common diagonalizing basis, i.e. an incoherent basis for the environment, which makes the whole unitary operator incoherent. We will review the definition of incoherent operations below.

We will assume that the system and the environment are initially in a product state, i.e.  $\rho_{SE} (0)= \rho_S \otimes \rho_E$. This assumption is justified here by the fact that we are interested in the build-up of correlations, thus we want to initiate the system in a completely uncorrelated state, so that any produced correlations are due to the interaction. Applying the unitary operator defined in \eqref{eq:unitary}, we obtain the following state $\sigma_{SE}$ after the interaction:
\begin{align}\label{eq:stateafter}
     \sigma_{SE} = U(\rho_S \otimes \rho_E)U^\dagger = \sum_{i,j}\op{i}{i}\rho_S\op{j}{j} \otimes U^i_E \rho_E {U^j_E}^\dagger
\end{align}
In order to evaluate the amount of $SE$ entanglement, we will use the relative entropy of entanglement $E_r$~\cite{Vedral-1997a,Vedral-1998a}, defined as:
\begin{align}\label{Er}
    E_r(\rho_{SE}) = \min_{\xi_{SE} \in \textrm{Sep}} S(\rho_{SE}||\xi_{SE})
\end{align}
where $S(\rho||\xi) = \Tr[\rho (\log \rho - \log \xi)]$ is the quantum relative entropy and the minimization is taken over all separable states $\xi_{SE} = \sum_i p_i \xi_S^i \otimes \xi_E^i$. Throughout this paper the base used in the logarithm is 2. 
This is a proper measure of entanglement~\cite{Plenio-2007a}, with a clear operational interpretation as the exponential bound on  asymptotic probability of error of confusing $\rho_{AB}$ with the nearest separable state via POVM measurements, see e.g. \cite{Vedral-2002a}.
The minimization in definition is the reason why it is difficult to evaluate the value directly. We employ the following lemma to bypass this problem:
\begin{lemma}\cite{Plenio-2000a} \label{lmm:entaglemenbound} The relative entropy of entanglement is bounded from below by the negative conditional entropy, i.e. for any bipartite state $\rho_{AB}$ we have:
\begin{equation}
    E_r(\rho_{AB}) \geq  \max\{S(\rho_A) - S(\rho_{AB}), S(\rho_B) - S(\rho_{AB}) \}
\end{equation}
\end{lemma}
\noindent where $S(\rho)$ is the von Neumann entropy: $S(\rho)= -\Tr[\rho \log \rho]$\blue{\cite{Nielsen-2000a}}. 

We will also need elements of the quantum coherence theory; see e.g. ~\cite{Baumgratz-2014a}. Since quantum coherence is a basis dependent resource, we need to specify which basis is incoherent. In principle, one can develop a theory with any incoherent basis. However, since we study a concrete family of open dynamics, given by \eqref{eq:Hamiltonian}, the pointer basis of the central system $\ket{i}$ is the natural candidate here because if the system loses coherence due to the interaction, decoheres to this basis. We do not specify the incoherent basis for environment because the coherence in the environment does not matter in our analysis as we will show below. In the language of coherence theory, an incoherent state is simply a diagonal state in the incoherent basis, which is of course nothing else but a decohered state in the open systems/quantum foundations terminology.
In a resource theory, of which coherence theory is an example, an operation which does not create resource is called a free operation with respect to that resource. Various free operations for coherence are proposed for different physical motivations~\cite{Chitambar-2016b}. The most standard one is incoherent operation whose Kraus operators $\{K_n\}$ satisfy that $\frac{K_n \rho K_n^\dagger}{\Tr[K_n \rho K_n^\dagger]}$ is incoherent for all $n$ and any incoherent state $\rho$. One can now define the relative entropy of coherence $C_r$ in analogy to the relative entropy of entanglement \eqref{Er} \cite{Baumgratz-2014a}:
\begin{align}\label{Cr}
    C_r(\rho) = \min_{\sigma \in \mathcal{I}}S(\rho||\sigma) = S(\Delta(\rho)) - S(\rho)
\end{align}
where $\Delta(\rho)$ is called the dephasing channel, defined as $\Delta(\rho) = \sum_i \op{i}{i}\rho \op{i}{i}$. 

Another ingredient that we will need is a measure of state distinguishability. Out of the possible measures \cite{Fuchs-1997a}, we will use the state fidelity $B(\rho, \sigma)$, also called generalized overlap, defined as:
\begin{align}\label{B}
    B(\rho,\sigma) =& \min_{\textrm{POVM}:\{M_m\}} \sum_m \sqrt{\Tr[\rho M_m] \Tr[\sigma M_m]} \nonumber \\ =& \Tr\sqrt{\rho^\frac{1}{2}\sigma \rho^\frac{1}{2}}
\end{align}
where $\Tr[\rho M_m]\Tr[\sigma M_m]$ is the probability that one obtains the same value $m$ from the measurement on $\rho$ and $\sigma$. Vanishing state fidelity is equivalent to a perfect distinguishability of the two states by POVM's. This justifies the use of the state fidelity as a measure of distinguishability. This measure does not have as clear operational meaning as e.g. the quantum Chernoff bound \cite{Audenaert-2007a}, related similarly to the quantum relative entropy to the asymptotic discrimination error of two states, but it is much easier to work with due to the compact form \eqref{B} and it also provides both an upper and a lower bound for the quantum Chernoff information.

Discriminating states from an ensemble $\mathcal{E}=\{p_i, \rho_i \}_i$ is a more difficult task and its efficiency can be quantified by the, so-called, accessible information $I_{acc}(\mathcal{E})$. It is defined as the maximum of the Shannon (classical) mutual information $I(\mathcal{I}:\mathcal{M})$ \cite{Cover-2006a} between the variable $\mathcal{I}$ for the ensemble index $i$ and the variable $\mathcal{M}$ for the measurement result $m$, taken over all POVM's $\{M_m\}$ on $\rho = \sum_i p_i \rho_i$. Recall that mutual information is defined as $ I(X:Y) = H(X) - H(X|Y)$ where $H(X)$ is the Shannon entropy and $H(X|Y)$ is the conditional entropy for random variable $X$ and $Y$~\cite{Cover-2006a}. The formal expression of accessible information is given by: 
\begin{align}\label{Iacc}
    I_{acc} =& \max_{POVM:\{M_m\}} I(\mathcal{I}:\mathcal{M}) = \max_{POVM:\{M_m\}} H(\mathcal{I}) - H(\mathcal{I}|\mathcal{M}),
\end{align}
where the conditional entropy $H(\mathcal{I}|\mathcal{M})$, associated with the POVM $\{M_m\}$ is defined as: 
\begin{align}
H(\mathcal{I}|\mathcal{M}) = -\sum_{i,m}p_i\Tr[\rho_iM_m]\log \frac{p_i\Tr[\rho_iM_m]}{\Tr[\sum_j p_j\rho_j M_m]}.
\end{align}

For two-element ensembles, the two notions of discrimination are connected by the following Lemma, which will be central for our reasoning: 
\begin{lemma}\cite{jain-2007a}\label{lmm:lemma2}[Lemma 3.6] Let $\mathcal{E} = \{\{(p,\rho_0\}, \{(1-p),\rho_1 \}\}$ be a two-element ensemble of quantum states of arbitrary finite dimensions. Then:
\begin{equation}
    I_{acc} (\mathcal{E}) \geq h(p) - 2\sqrt{p(1-p)}B(\rho_0,\rho_1)
\end{equation}
    where $h$ is the binary Shannon entropy: $h(p) = -p\log p -(1-p)\log (1-p)$.
\end{lemma}


\section{Bounds on the system-environment entanglement generation}\label{Sec:Results}
In the pure dephasing, open evolution, described in the Introduction, the system's information is transferred to the environment through the controlled-unitary interaction \eqref{eq:unitary}. This interaction  can be interpreted as the environment implementing a measurement in the pointer basis $\{\ket{i}\}$  and storing the information in its state. Indeed, as one can easily prove, the state of the environment in the course of the evolution becomes:
\begin{align}\label{sigmaE}
\sigma_E(t) = \Tr_S[\sigma_{SE}(t)] = \Tr_S[U(t)\rho_S\otimes \rho_EU(t)^\dagger] = \sum_i p_i \rho_i(t), 
\end{align}
where  $p_i=\bra{i}\rho_S\ket{i}$ are the initial pointer probabilities and the conditional states of the environment $\rho_i(t)=U_E^i(t)\rho_E{U_E^i}(t)^\dagger$ encode, in a generally imperfect way, the information about the state of the central system.
One can now try to recover this information by implementing a POVM measurement $\{M_m\}$ on the environment in order to extract the value of $\mathcal{I}$, encoded by the environment in the course of the evolution. One may view this as an indirect measurement of the system through the environment, common in the macroscopic world; cf. the studies of objectivity in quantum mechanics in e.g. \cite{Zurek-2009a,Korbicz-2021a}.
The uncertainty of this indirect measurement is represented by Shannon conditional entropy $H(\mathcal{I}|\mathcal{M})$ \cite{Cover-2006a} as $\mathcal{I}$ represents the variable for index on the right hand side of \eqref{sigmaE} and $\mathcal{M}$ represents the variable for measurement value. From above observations, it is natural to think the less this uncertainty is, the stronger the system and the environment are correlated. On the other hand, the more is there initial coherence in the pointer basis $\{\ket{i}\}$, the more entanglement should be generated \cite{Killoran-2016a}. Indeed this is so as the following Lemma shows:
\begin{lemma}\label{thm:main}
     For any POVM measurement $\{M_m\}$ on the environment, the relative entropy of entanglement, generated by the pure dephasing evolution $U$, is bounded from below by:
    \begin{align}\label{main_bound}
        E_r(\sigma_{SE}) \geq C_r(\rho_S) - H(\mathcal{I}|\mathcal{M}).
    \end{align}
\end{lemma}
\begin{proof}
    Applying Lemma~\ref{lmm:entaglemenbound} to $\sigma_{SE}$, we have
%
\begin{align}
    E_r(\sigma_{SE}) \geq  \max \{S(\sigma_S) - S(\sigma_{SE}), S(\sigma_E) - S(\sigma_{SE}) \},  
\end{align}
where $\sigma_S = \Tr_E[\sigma_{SE}]$. We focus on finding a lower bound for the latter, $S(\sigma_E) - S(\sigma_{SE})$, because $S(\sigma_E)$ can be transformed to the conditional entropy or the overlap of environmental states through lemma \ref{lmm:lemma2}, leading us to a nice result represented in corollary \ref{thm:main2}.
We first observe that: 
\begin{align}\label{eq:firstbound}
    S(\sigma_E) - S(\sigma_{SE}) = S(\sigma_E) - S(\rho_S) - S(\rho_E)
\end{align}
where we used fact that the unitary evolution preserves the von Neumann entropy and it is additive for product states. Notice \eqref{eq:firstbound} can be negative while the relative entropy of entanglement should be always non-negative.

Next, to find a lower bound of $S(\sigma_E)$, let us recall the Holevo bound~\cite{Holevo-1973a} on the $I_{acc}$ associated with \eqref{sigmaE}: 
\begin{align}
I_{acc}(\sigma_E) \leq S(\sigma_E) - \sum_i \bra{i}\rho_S\ket{i} S(U_i\rho_EU_i^\dagger),
\end{align}  
from which it follows that:
\begin{align}\label{eq:ineq1}
   S(\sigma_E) \geq I_{acc}(\sigma_E) + S(\rho_E),
\end{align}
since $ S(U_i\rho_EU_i^\dagger)=S(\rho_E)$.
From the definition \eqref{Iacc}, $I_{acc}(\sigma_E)$ is also bounded by $H(\mathcal{I}) - H(\mathcal{I}|\mathcal{M})$. $\mathcal{I}$ is distributed with the initial pointer probabilities $p_i$, as follows from \eqref{sigmaE}, or equivalently with the eigenvalues of the fully decohered state of system $\Delta(\rho_S)=\sum_i \op{i}{i}\rho_S \op{i}{i}$. 
Therefore we may write $H(\mathcal{I}) = S(\Delta(\rho_S))$. Since  there is a maximization over all POVM in the definition \eqref{Iacc}, for an arbitrary POVM $\{M_m\}$ we will have:
\begin{align}\label{eq:ineq2}
    I_{acc}(\sigma_E) \geq H(\mathcal{I}) - H(\mathcal{I}|\mathcal{M}) = S(\Delta(\rho_S)) - H(\mathcal{I}|\mathcal{M})
\end{align}

Now we recover the statement by substituting \eqref{eq:ineq1} and \eqref{eq:ineq2} into \eqref{eq:firstbound}:
\begin{align}
    E_r(\sigma_{SE}) \geq & S(\sigma_E) - S(\rho_S) - S(\rho_E) \nonumber\\
    \geq & I_{acc} - S(\rho_S) \label{2ndstep} \\
    \geq & S(\Delta(\rho_S)) - H(\mathcal{I}|\mathcal{M}) - S(\rho_S)  \nonumber \\ 
    =& C_r(\rho_S) - H(\mathcal{I}|\mathcal{M}),\nonumber
\end{align}
where in the final line we used the fact \eqref{Cr}. \end{proof}

The above Lemma is valid for any finite dimensionality of the central system and the environment. The price to pay for this generality is the presence of the arbitrary POVM in the bound \eqref{main_bound}.
This arbitrariness can be removed when the central system is a qubit. Then, their ensemble \eqref{sigmaE} consists of only two states $\rho_0(t)=U_0(t)\rho_EU_0(t)^\dagger$, $\rho_1(t)=U_1(t)\rho_EU_1(t)^\dagger$ with probabilities $p = \bra{0}\rho_S\ket{0}$ and $1-p = \bra{1}\rho_S\ket{1}$. We can then apply the Lemma~\ref{lmm:lemma2} and obtain:
\begin{corollary}\label{thm:main2}
If the central system is two-dimensional, then:
     \begin{align}\label{eq:thm}
       E_r(\sigma_{SE}) \geq  C_r(\rho_S) - 2\sqrt{p(1-p)}B(\rho_0(t),\rho_1(t)).
    \end{align}
\end{corollary}
The proof follows immediately from the application of Lemma~\ref{lmm:lemma2} in \eqref{2ndstep}.

Let us look closely at Lemma \ref{thm:main} and give some interpretations. The second term $H(\mathcal{I}|\mathcal{M})$ on the right hand side is the uncertainty about the system's state measured by the environment after our best measurement on environment. Non zero $H(\mathcal{I}|\mathcal{M})$ means that we cannot conclude the system's value in the pointer basis even if we could use any POVM on the environment. 
Conversely, if $H(\mathcal{I}|\mathcal{M})$ becomes zero, the information about the system is transferred fully to environment and entanglement must be there unless system's coherence is zero, which means the system is fully classical state. After all, when the system has more coherence in pointer bases, more entanglement is required to know the value of system certainly. One can also re-arrange the inequality \eqref{main_bound} and obtain the lower bound on the uncertainty about the system's state in terms of initial coherences and the generated entanglement:
\begin{align}
 \max_{POVM:\{M_m\}} H(\mathcal{I}|\mathcal{M})\geq C_r(\rho_S)  - E_r(\sigma_{SE}). 
\end{align}

Similar interpretation can be applied to the qubit version \eqref{eq:thm}. If $p(1-p) = 0$ one knows the system's state with certainty. Notice that $C_r(\rho_S) = 0$ in this case. Also if $B(\rho_0(t),\rho_1(t)) = 0$, it implies that $\rho_0(t)$, $\rho_1(t)$ are perfectly distinguishable, which means maximum information has been transferred to the environment. 

Note that the bounds \eqref{main_bound},\eqref{eq:thm} are not tight and can even become negative as we have mentioned in the proof and will see in an example in the next section. In that case of course the bounds trivialize as $0\leq E_r$ always.  Nevertheless, despite the lack of optimality, they provide a valuable information about the open evolution, certifying that at least a given amount of entanglement has been generated. This is particularly interesting when the right hand sides of  \eqref{main_bound},\eqref{eq:thm} are large, which implies a large amount of created entanglement.

Finally, we will give an upper bound on the generated entanglement. From a previous research~\cite{Streltsov-2015a}, an upper bound on entanglement generated by incoherent operations with the initially coherent system $\rho_S$ and an incoherent pure ancilla $\op{0}{0}_E$ is given by the initial coherence in system, for any measure of entanglement $E_D$ and coherence $C_D$ corresponding to a contractive divergence $D$: 
\begin{align}
    E_D(\Lambda[\rho_S \otimes \op{0}{0}_E]) \leq C_D(\rho_S) 
\end{align}
where $\Lambda$ is any incoherent operation on the bipartite system. When we specify $D$ as the relative entropy, which we have been using throughout this work, the above upper bound becomes tight. 

However, since the corresponding unitary evolution $U$ for the given Hamiltonian of \eqref{eq:Hamiltonian} is not incoherent in general regardless of incoherent basis for environment and we assume the environmental state arbitrary, the above upper bound is not valid in our scenario. Also, considering our intuitive picture of system-environment interaction where coherence is converted into entanglement, it is interesting to connect the generated entanglement not only to the initial coherence $C_r(\rho_S)$ as we did in lemma~\ref{thm:main} for lower bound, but to the decoherence $C_r(\rho_S) - C_r(\sigma_S)$, the difference of coherence between the initial and final states in system as we demonstrate below.

\begin{lemma}\label{lem:sandwich}
    The relative entropy of entanglement generated by pure dephasing evolution $U$ is bounded from above as:
    \begin{align}
        E_r(\sigma_{SE}) \leq C_r(\rho_S) - C_r(\sigma_S) + H(\mathcal{I}).
    \end{align}
    \begin{proof}
From the definition of the relative entropy of entanglement, $E_r(\sigma_{SE}) \leq S(\sigma_{SE}||\sigma_S \otimes \sigma_E)$. 
We observe that $\log \sigma_S \otimes \sigma_E =\log \sigma_S \otimes \mathbbm{I} + \mathbbm{I} \otimes \log \sigma_E $ and $\Tr[\sigma_{SE}\log \sigma_S \otimes \mathbbm{I}] = \Tr_S[\Tr_E[\sigma_{SE}] \log \sigma_S ] = \Tr_S[\sigma_S \log \sigma_S] = -S(\sigma_S).$ Then we have
\begin{align}
     E_r(\sigma_{SE}) \leq & \ S(\sigma_{SE}||\sigma_S \otimes \sigma_E) \nonumber \\ =& -S(\sigma_{SE}) - \Tr[\sigma_{SE}\log \sigma_S \otimes \sigma_E] \nonumber \\   =& -S(\rho_S) - S(\rho_E) - S(\sigma_S) + S(\sigma_E)
\end{align}
We also observe that $S(\sigma_E) \leq S(\rho_E) + H(\mathcal{I})$, which follows from the well known property of the von Neumann entropy: $S(\sum_i p_i \rho_i) \leq  \sum p_i S(\rho_i) + H(\mathcal{I})$. 
Substituting this into the above inequality, we prove the statement. 
    \end{proof}
\end{lemma}

Lemma \ref{lem:sandwich} fits the intuition that larger system-environment entanglement leads to more decoherence as quantified by $D_S$. More precisely, we show that $D_S$ is lower bounded by the generated entanglement:
\begin{align}
    E_r(\sigma_{SE}) -H(\mathcal{I}) \leq C_r(\rho_S) - C_r(\sigma_S). 
\end{align}


\section{Case study: The spin-boson model}
In this Section, we illustrate our results on the spin-boson model, which is one of the most popular open systems models; see e.g. \cite{Schlosshauer-2007a,Breuer-2002a}, and the references therein for an introduction. In its most used form, it describes a two-level central system, a qubit, interacting with a collection of bosonic modes (or harmonic oscillators), constituting the environment;  for an example of higher dimensional central systems see e.g. \cite{Kicinski-2021a}. The dynamics of the model is given by the following Hamiltonian:
\begin{align}\label{Hsb}
H&=   \sum_{\boldsymbol{k}}\omega_{\boldsymbol{k}} a^{\dagger}_{\boldsymbol{k}} a_{\boldsymbol{k}} +  \sigma_z \otimes \sum_{\boldsymbol{k}}\left(g_{\boldsymbol{k}} a^\dagger_{\boldsymbol{k}} + g^{*}_{\boldsymbol{k}} a_{\boldsymbol{k}}\right)  
\end{align}
where for the purpose of this example, we have suppressed the self-Hamiltonian of the central system, or assumed it commutes with the interaction, as otherwise the evolution would not be of the pure dephasing form \eqref{eq:unitary}. Here, $a^{\dagger}_{\boldsymbol{k}}$, $a_{\boldsymbol{k}}$ are the annihilation and creation operators of the mode $\boldsymbol{k}$ of the environment, $\sigma_z$ is the $z$ axis Pauli matrix acting in the space of the central system and $g_{\boldsymbol{k}}$ are the coupling constants. The pointer basis here is the eigenbasis of $\sigma_z$, which follows immediately from the interaction term in  \eqref{Hsb}.  

We assume the environment to be thermal at temperature $T$, i.e. $\rho_E=e^{-\beta H_E}/Z$, where we choose the units so that $\beta =1 /T$, and the central system to be initially in the following mixed state:
\begin{align}\label{rhoS}
    \rho_S = (1-\alpha) \op{+}{+} + \alpha \frac{\mathbbm{I}}{2}.
\end{align} 
It is a mixture of a state that is maximally coherent w.r.t. the pointer basis   and the maximally mixed state $\mathbbm{I}/2$, and  $\alpha$ is the mixedness parameter, $0\leq \alpha\leq 1$. Such a choice of the initial state will allow us to control the initial coherence in the pointer basis.

Since the central system is two-dimensional, we will use the bound \eqref{eq:thm}. The initial coherence  $C_r(\rho_S)$ can be easily calculated and is given by the binary entropy:
\begin{align}
C_r(\rho_S) = 1 - h\left(\frac{\alpha}{2}\right)
\end{align}
Interestingly, the fidelity $B(\rho_0(t),\rho_1(t))$ can also be calculated analytically for the popular Ohmic spectral density:
\begin{align}
J(\omega) = \frac{\omega^s}{\Lambda^{s-1}}\e^{-\omega/\Lambda},
\end{align}
where $\Lambda$ is the cut-off frequency and $s$ is the Ohmicity parameter. The value  $s=3$ is a known threshold, distinguishing Markovian evolution for $s<3$ from non-Markovian for $s\geq 3$; see e.g. \cite{Addis-2014a}. In \cite{Lampo-2017a} the analytical expression for $B$ was obtained for the case of the Ohmicity parameter $s>1$. It consists of two parts, vacuum, i.e. temperature-independent, and thermal: 
\begin{widetext}
\begin{align}
    \ln B(t) =& \ln B_{vac} (t)+ \ln B_{th}(t),\\
    \frac{1}{2} \ln B_{vac}(t) =& -\Gamma(s-1) \left [ 1 - \frac{\cos [(s-1)\arctan(\Lambda t)]}{(1+\Lambda^2 t^2)^{\frac{s-1}{2}}} \right ], \\
    \frac{1}{4}\ln B_{th}(t) =& \left ( -\frac{T}{2\Lambda} \right )^{s-1} \bigg \{\Psi^{s-2} \left( 1 + \frac{T}{2\Lambda} \right ) - \Psi^{s-2}\left ( \frac{1}{2} +\frac{T}{2\Lambda} \right ) 
    + \left [ \frac{1}{2}\Psi^{s-2}\left ( \frac{1}{2} + \frac{T}{2\Lambda} +i\frac{Tt}{2} \right )
    - \frac{1}{2}\Psi^{s-2}\left ( 1 + \frac{T}{2\Lambda} +i\frac{Tt}{2}\right ) +\textrm{c.c.} \right ] \bigg \},
\end{align}
\end{widetext}
where c.c. means complex conjugate of the two terms in the square parenthesis, $\Gamma(s)$ is the Euler gamma function, and $\Psi^m(z)$ is the polygamma function defined as: 
\begin{align}
    \Psi^m(z) = \sum_{k=0}^\infty \frac{(-1)^{m+1}m!}{(z+k)^{m+1}}.
\end{align}

\begin{figure}
    \begin{subfigure}{\columnwidth}
    \includegraphics[width = \columnwidth]{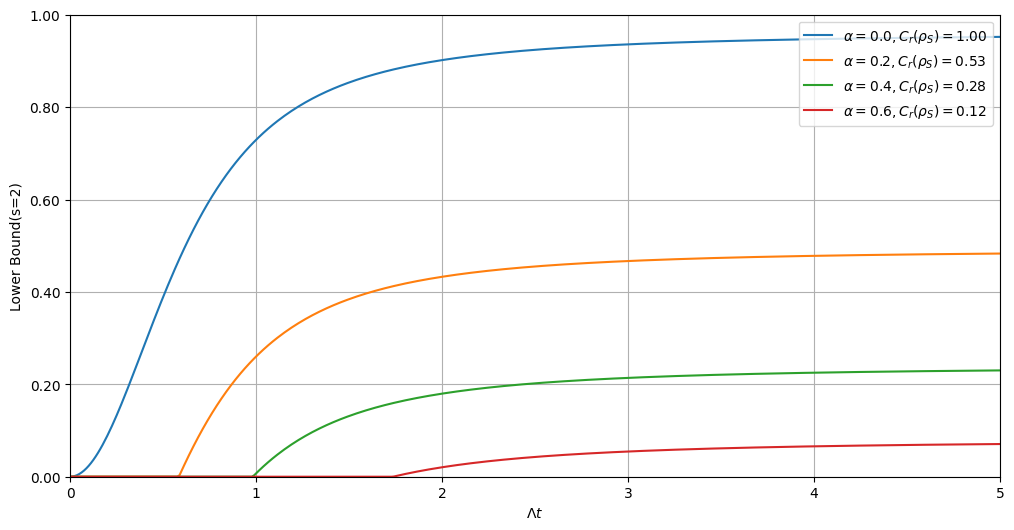}\caption{}
    \end{subfigure}
    \begin{subfigure}{\columnwidth}
    \includegraphics[width = \columnwidth]{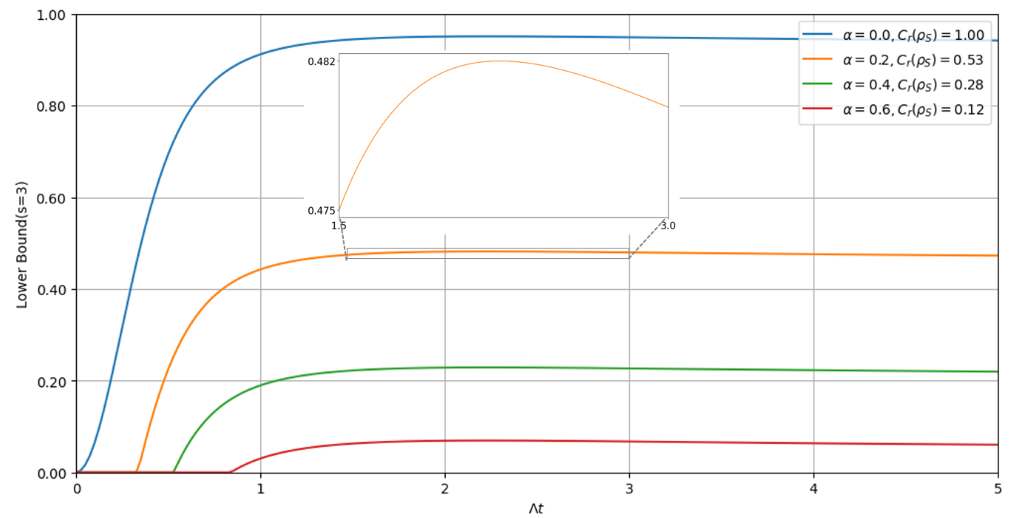}\caption{}
    \end{subfigure}
    \caption{Plots of the lower bound \eqref{eq:thm} for the spin-boson model as a function of time for a different mixedness parameter $\alpha$ (cf. \eqref{rhoS}) and different Ohmicity. The regions where the right-hand side of \eqref{eq:thm} becomes negative, i.e. the bound trivializes, have been substituted by zero, i.e. the actual plotted function is $\max\{0, RHS \eqref{eq:thm}\}$, which is responsible for the non-analytical points on the plots. a) Ohmicity $s=2$, corresponding to a Markovian evolution; b) $s=3$ corresponding to a non-Markovian evolution. The ratio of the thermal to the cutoff energy is  set to $T/\Lambda = 1$ and the time is measured in the cutoff units. The inset shows the peak and thus the non-monotonic behavior of the plots (with different vertical scale for each plot). 
    }
    \label{fig:fixed ratio}
\end{figure}
\begin{figure}
    \begin{subfigure}{\columnwidth}
        \includegraphics[width = \columnwidth]{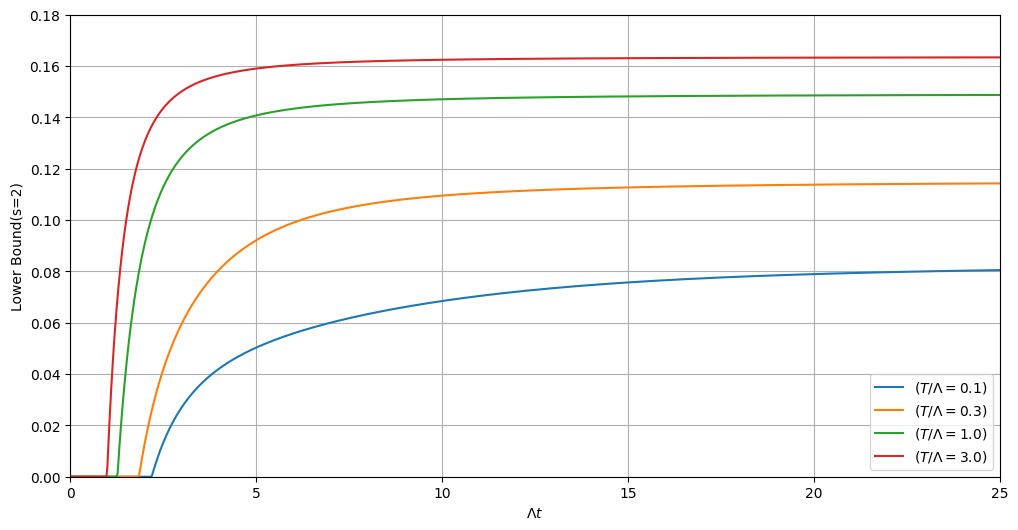}\caption{}
    \end{subfigure}
    \begin{subfigure}{\columnwidth}
        \includegraphics[width = \columnwidth]{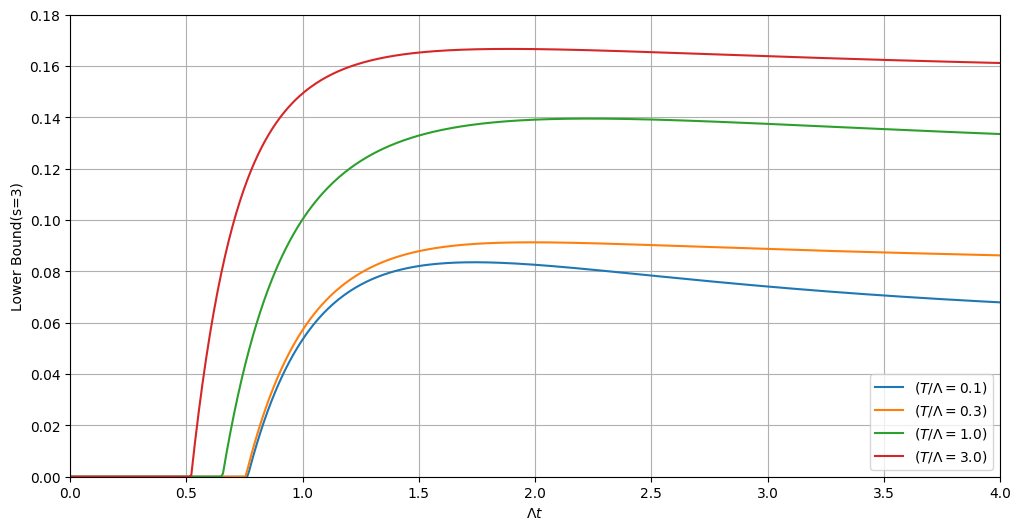}\caption{}
    \end{subfigure}
    \caption{The same as in Fig.~\ref{fig:fixed ratio} but for a fixed mixedness parameter $\alpha = 1/2$ and different ratios $T/\Lambda$. a) Ohmicity $s=2$; b) Ohmicity $s=3$.  }
    \label{fig:fixed coherence}
\end{figure}

The resulting behavior of the lower bound is presented in Figs.~\ref{fig:fixed ratio}, \ref{fig:fixed coherence}. The regions in the plots where the right-hand side of  \eqref{eq:thm} becomes negative, i.e. the bound trivializes since by definition $E_r \geq 0$ always, were substituted by zero to: i) avoid a suggestion that $E_r$ might be negative; ii)  emphasize what might be called "a sudden revival of the bound". Of course, since this is a lower bound, the fact that in some situations the bound trivializes does not mean there cannot be any entanglement and thus the discontinuity for the initial times need not mean a sudden birth of entanglement. 
But on the other hand,  it is intuitively appealing that since we initiate the whole model in a product state, it will take some time for the central system and the environment to get entangled.  

In Fig.~\ref{fig:fixed ratio}, we keep fixed the ratio of the thermal to the cutoff energy, $T/\Lambda = 1$, and plot \eqref{eq:thm} for different mixedness $\alpha$ and Ohmicity $s$ parameters. 
As expected from the \eqref{eq:thm}, less mixed, i.e. more coherent, states with lower $\alpha$ lead to much more entanglement generation. More interestingly, there is a slight quantitative change of behavior between $s=2$ and $s=3$, corresponding to the transition between the Markovian and non-Markovian evolution respectively. While for $s=2$ the behavior is monotonic in time, for $s=3$ there is a broad peak around $\Lambda t = 2$, regardless of the value of $\alpha$ (depicted on the inset). 
This peak is more visible in Fig.~\ref{fig:fixed coherence}, where we fix the mixedness $\alpha = 1/2$ and vary the energy ratio $T/\Lambda$. The non-monotonic behavior of the entanglement bound for $s=3$ bears certain similarity to non-Markovianity measures, where a non-monotonic behavior  of some marker function certifies information back-flow and hence the non-Markovianity of the evolution. Especially close seems the measure from \cite{Hall-2014a}, where the non-Markovianity is detected via non-monotonicity of the decoherence rate, which in turn derives from system-environment correlations.  

We also notice from Fig.~\ref{fig:fixed ratio} that the higher the $T/\Lambda $ ratio, the higher the entanglement bound, which means that if the cutoff frequency $\Lambda$ is fixed, higher temperatures guarantee stronger entanglement between the system and the environment. This in turn leads to a stronger decoherence and explains the well known fact why the system need to be kept cold in order to maintain the coherence.

\section{Conclusions}
We studied generation of entanglement between the system of interest and its environment for a broad class of so-called pure dephasing evolutions, also known as measurement-type interactions. In this part of the work, we extended previous results ~\cite{Roszak-2015a, Roszak-2018a} obtained for pure initial states of the system, to the case of mixed states. We combined methods and results from coherence theory, state discrimination, and open systems evolution to obtain a series of bounds on the amount of the generated entanglement as measured by the relative entropy. 
Firstly, and most importantly, we obtained a lower bound, which certifies that at least the given amount of entanglement was created during the evolution out of a given amount of the initial coherence. The lower bound character of this result is a price to pay for the generalization to arbitrary, mixed states of the system. For a qubit as the central system, we obtained a closed expression that connects the amount of entanglement to the state distinguishability of the environmental states. This makes an interesting link to certain studies of quantum Darwinism \cite{Korbicz-2021a}. For an arbitrary dimension, our bound still contains a certain optimization. As an example, we calculated analytically and numerically the lower bound for the popular spin-boson system, rediscovering well-known traces of the non-Markovianity in the entanglement dynamics. Clearly, more systems can be analyzed, especially since for some of them the relevant state distinguishability has already been quantified \cite{Korbicz-2014a,Tuziemski-2019a,Kwiatkowski-2021a,Kicinski-2021a}.

Next, we derived an upper bound for the entanglement in terms of the difference of coherence in the system before and after the evolution, which is a measure of the decoherence efficiency. This proves in an elegant and general way that stronger entanglement leads to more decoherence and vice versa, more decoherence is accompanied by more entanglement generation. This bound can also be seen as an extension of the result obtained in \cite{Streltsov-2015a}.

Summarizing, our results allow us to better characterize and understand the correlations that pure dephasing interaction can produce under general initial conditions. Because of the ubiquitous character of the pure dephasing interaction in open systems \cite{Schlosshauer-2007a,Breuer-2002a}, our results are fairly universal across the field of open systems and decoherence theory. Moreover, since quantum measurement theory is largely based on the same type of interaction, we characterize how much entanglement is generated during a measurement process, thus linking to the foundations of quantum mechanics.  On the other hand, a possible practical importance of our studies may lie in the development of non-standard decoherence-mitigation techniques, which rely on the system-environment correlations (see e.g. \cite{Roszak-2023a} for an example of such a technique).

\section*{Acknowledgements}
We would like to thank K. Roszak for getting us interested in the topic and for fruitful discussions and comments. We acknowledge the financial support of the Polish National Science Foundation (NCN) through the grants Opus 2019/35/B/ST2/01896 and Quant-Era ``Quantum Coherence Activation By Open Systems and Environments'' QuCABOoSE 2023/05/Y/ST2/00139.

\bibliography{references}

\end{document}